\documentclass[11pt,reqno]{amsart}

\usepackage{amssymb,amsmath,amsthm,amscd,latexsym,amsfonts}
\usepackage{mathtools}
\usepackage[T1]{fontenc}
\usepackage{graphicx,epstopdf}
\usepackage{tikz}
\usepackage{graphicx}
\usepackage{xcolor}
\usepackage{cite}
\usepackage{MnSymbol}
\usepackage[a4paper,top=3cm, bottom=3cm, left=3.1cm, right=3.1cm]{geometry}

\newtheorem{thm}{Theorem}

\newtheorem{lemma}{Lemma}
\newtheorem{pro}{Proposition}
\newtheorem{rk}{Remark}
\newtheorem{cor}{Corollary}

\numberwithin{equation}{section} 

\newcommand{\bea}{\begin{eqnarray}}
	\newcommand{\eea}{\end{eqnarray}}

\begin{document}
	\title [Gibbs measures for a HC-SOS model]
	{Extreme Gibbs measures for a Hard-Core-SOS model on Cayley trees}
	
\author {R.M. Khakimov, M.T. Makhammadaliev,  U.A. Rozikov}
	
\address{R.M. Khakimov$^{a,b}$, M.T. Makhammadaliev$^{a,b}$
		\begin{itemize}
			\item[$^a$] V.I.Romanovskiy Institute of Mathematics, Uzbekistan Academy of Sciences, 9, Universitet str., 100174, Tashkent, Uzbekistan;
			\item[$^b$] Namangan State  University, Namangan, Uzbekistan.
	\end{itemize}}
	\email{rustam-7102@rambler.ru, mmtmuxtor93@mail.ru}
	
	\address{ U.Rozikov$^{a,c,d}$
		\begin{itemize}
			\item[$^a$] V.I.Romanovskiy Institute of Mathematics, Uzbekistan Academy of Sciences, 9, Universitet str., 100174, Tashkent, Uzbekistan;
			\item[$^c$] National University of Uzbekistan,  4, Universitet str., 100174, Tashkent, Uzbekistan.
		 \item[$^d$]	Qarshi State University, 17, Kuchabag str., 180119, Qarshi, Uzbekistan.
	\end{itemize}}
	\email{rozikovu@yandex.ru}
	
	\begin{abstract}
			
		We investigate splitting Gibbs measures (SGMs) of a three-state (wand-graph) hardcore SOS model on Cayley trees of order $ k \geq 2 $. Recently, this model was studied for the hinge-graph with $ k = 2, 3 $, while the case $ k \geq 4 $ remains unresolved. It was shown that as the coupling strength $\theta$ increases, the number of translation-invariant SGMs (TISGMs) evolves through the sequence $ 1 \to 3 \to 5 \to 6 \to 7 $.
		
		In this paper, for wand-graph we demonstrate that for arbitrary $ k \geq 2 $, the number of TISGMs is at most three, denoted by $ \mu_i $, $ i = 0, 1, 2 $. We derive the exact critical value $\theta_{\text{cr}}(k)$ at which the non-uniqueness of TISGMs begins. The measure $ \mu_0 $ exists for any $\theta > 0$.
		
		Next, we investigate whether $ \mu_i $, $i=0,1,2$ is extreme or non-extreme in the set of all Gibbs measures.
		
		 The results are quite intriguing:
		 
		1) For $\mu_0$:
		
		- For $ k = 2 $ and $ k = 3 $, there exist critical values $\theta_1(k)$ and $\theta_2(k)$ such that $ \mu_0 $ is extreme if and only if $\theta \in (\theta_1, \theta_2)$, excluding the boundary values $\theta_1$ and $\theta_2$, where the extremality remains undetermined.
		
		- Moreover, for $ k \geq 4 $, $ \mu_0 $ is never extreme.
		
		2) For $\mu_1$ and $\mu_2$ at $k=2$ there is $\theta_5<\theta_{\text{cr}}(2)=1$ such that these measures are extreme if $\theta \in (\theta_5, 1)$. 
		
	\end{abstract}
	\maketitle
	

	
	\section{Introduction}

The existence of Gibbs measures for a broad class of Hamiltonians was first established in Dobrushin's seminal work (see, e.g., \cite{G,R,Rb, Si, Ve}). However, fully characterizing the set of Gibbs measures for a given Hamiltonian remains challenging.  

At high temperatures, Gibbs measures are typically unique (see \cite{G, Pr, Si}), reflecting the absence of phase transitions. In contrast, low-temperature analysis often requires specific assumptions about the Hamiltonian. For continuous Hamiltonians, Gibbs measures form a nonempty, convex, compact set in the space of probability measures (see Chapter 7 in \cite{G}). The extreme points of this set, called extremal measures, correspond to pure phases and belong to the class of splitting Gibbs measures (see Chapter 11 in \cite{G}). 

This paper focuses on the extreme points of the set of Gibbs measures for the Hard-Core-SOS model. The SOS (solid-on-solid) model, introduced on Cayley trees in \cite{RS} as a generalization of the Ising model, has been extensively studied (e.g., \cite{O1, KRS, O2}). Unlike other models, the Hard-Core (HC) model imposes constraints on spin values, with applications in combinatorics, statistical mechanics, and queuing theory. HC models are relevant for studying random independent sets on graphs \cite{bw1, gk} and gas molecules on lattices \cite{Ba}. Numerous works explore limiting Gibbs measures for HC models with finite states on Cayley trees (see \cite{R, KMU, BW, KhMH, RKhM}). 

Here, we study translation-invariant splitting Gibbs measures (TISGMs) for Hard-Core-SOS models with activity $\lambda > 0$ on Cayley trees of order $k \geq 2$. In \cite{BW}  four specific models - wrench, wand, hinge, and pipe - are considered. In \cite{BR}, this model was analyzed with a hinge admissibility graph. 

This paper extends these studies to other configurations. In Section 2 we give main definitions and known facts. In Section 3  for wand-graph we demonstrate that for arbitrary $ k \geq 2 $, the number of TISGMs is at most three, denoted by $ \mu_i $, $ i = 0, 1, 2 $. We derive the exact critical value $\theta_{\text{cr}}(k)$ at which the non-uniqueness of TISGMs begins. In the last Section 4 we investigate whether $ \mu_i $, $i=0,1,2$ is extreme or non-extreme in the set of all Gibbs measures.

\section{\textbf{Definitions and known facts}}\
The Cayley tree $\Im^k$ of order $ k\geq 1 $ is an infinite tree, i.e. a cycles-free graph such that from each
vertex of which issues exactly $k+1$ edges. We denote by $V$ the set of the vertices of tree and by $L$ the
set of edges. The distance on the Cayley tree, denoted by $d(x,y)$, is defined as the number of nearest neighbor pairs of the shortest path between the vertices $x$ and $y$ (where path is a collection of nearest neighbor pairs, two consecutive pairs sharing at least a given vertex)

For a fixed $x^0\in V$, called the root, let
$$W_n =\{x\in V \ |\ d (x,x^0) =n \}, \qquad V_n = \bigcup_{m=0}^{n}{W_m}$$
be respectively the ball and the sphere of radius $n$ with center at $x^0$.
For $x\in W_n$ let
$$S(x)=\{y_i\in W_{n+1} \ | \  d(x,y_i)=1, i=1,2,\ldots, k \},$$
be the set of direct successors of $x$. Note that in $\Im^k$ any vertex $x\neq x^0$ has $k$ direct successors, and root $x^0$ has $k+1$ direct successors.

Next, we denote by $\Phi =\{0,1,2,\dots,m\}$ the {\it local state space}, i.e., the space of values of the spins
associated to each vertex of the tree. Then, a {\it configuration} on the Cayley tree is a collection
$\sigma=\{\sigma(x):x\in V\}\in \Phi^V=\Omega$.

Let us now describe hardcore interactions between spins of neighboring vertices. For this, let
$G=(\Phi,K)$ be a graph with vertex set $\Phi$, the set of spin values, and edge set $K$. A configuration $\sigma$
is called $G$-admissible on a Cayley tree if $\{\sigma(x), \sigma(y)\}\in K$ is an edge of $G$ for any pair of nearest
neighbors $\langle x,y\rangle \in L$. We let $\Omega^G$ denote the sets of $G$-admissible configurations. The restriction of
a configuration on a subset $A\subset V$ is denoted by $\sigma_A$ and $\Omega_A^G$ denotes the set of all $G$-admissible
configurations on $A$. On a general level, we further define the matrix of activity on edges of $G$ as a
function:
$$\lambda:\{i,j\}\in K\rightarrow \lambda_{i,j}\in\mathbb{R}_{+},$$
where $\mathbb{R}_{+}$ denotes the positive real numbers and $\lambda_{i,j}$ is called the activity of the edge $\{i,j\}\in K$.
In this paper, we consider the graph $G$ as shown in Fig. \ref{fi1}, which is called a wand-graph, see for
example \cite{BW}. In words, in the wand-graph $G$, configuration are admissible only if, for any pair of nearest-neighbor vertices $x, y$, we have that
$$|\sigma(x)-\sigma(y)|\in\{0,1\}.$$
Note that our choice of admissibilities generalizes certain finite-state random homomorphism models, see \cite{LT}, where only configurations with
$|\sigma(x)-\sigma(y)|=1$ are allowed.

\begin{figure}[h]
\center{\includegraphics[width=8cm]{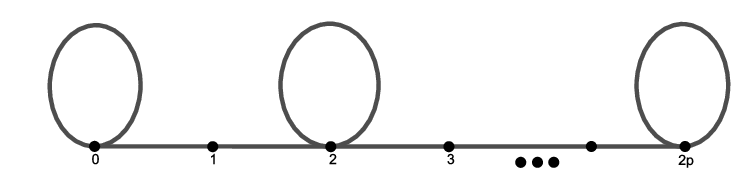}}
\caption{The wand-graph $G$ with $2p+1$ vertices, where $m=2p$.}
\label{fi1}
\end{figure}

Our main interest lies in the analysis of the set of {\it splitting Gibbs measures} (SGMs) defined on
wand-graph addmissible configurations. Let us start by defining SGMs for general admissibility
graphs $G$.

For given $G$ and $\lambda$ we define the Hamiltonian of the HC-SOS model as
$$
 H(\sigma)=
 \begin{cases}-J\sum_{\langle x,y\rangle}
|\sigma(x)-\sigma(y)|, \quad \mbox{if} \quad \sigma\in\Omega^G,\\
+\infty, \qquad \qquad \qquad \qquad \quad \mbox{if} \quad  \sigma\notin\Omega^G.
\end{cases}
$$

Let
$$z:x\rightarrow z_x=(z_{0,x},z_{1,x},\dots,z_{m,x})\in\mathbb{R}_{+}^{m+1}$$
be a vector-valued function on $V$. Then, given $n=1,2,\dots$ and an activity $\lambda=(\lambda_{i,j})_{\{i,j\}\in K}$,
consider the probability distribution $\mu(n)$ on $\Omega_{V_n}^G$, defined as
\begin{equation}\label{e1}
\mu_n(\sigma_n)=Z_n^{-1}\prod_{\langle x,y\rangle \in V_n}{\lambda_{\sigma_n(x),\sigma_n(y)}}\prod_{x \in W_n}{z_{\sigma(x),x}},
\end{equation}
where $\sigma_n=\sigma_{V_n}$. Here $Z_n$ is the partition function.

The probability distributions (\ref{e1}) are  compatible if for all $\sigma_{n-1}\in \Phi^{V_{n-1}}$ one has
\begin{equation}\label{e2}
\sum_{\omega_n\in \Phi^{W_n}}\mu_n(\sigma_{n-1}\vee \omega_n)=\mu_{n-1}(\sigma_{n-1}),
\end{equation}
where $\sigma_{n-1}\vee \omega_n$  is the concatenation of the configurations.
Under condition (\ref{e2}), by the well-known Kolmogorov's extension theorem, there exists a unique measure $\mu$ on $\Phi^V$, such that $\forall n\in\mathbb{N}$ and $\sigma_n\in\Phi^{V_n}$
$$
\mu(\{\sigma\mid_{V_n}=\sigma_n\})=\mu_n(\sigma_n).
$$
and we call it a \emph{splitting Gibbs measure} (SGM) corresponding to the activity  $\lambda$ and  vector-valued function $z_x, x\in V$.

Let $K$ be the set of edges of a graph $G$. We let $A\equiv A^G=\big(a_{ij}\big)_{\{i,j\}\in K}$ denote the adjacency
matrix of the graph $G$, i.e.,
$$a_{ij}=a_{ij}^G=%
\begin{cases} 1, \quad \mbox{if} \quad \{i,j\}\in K, \\
0, \quad \mbox{if} \quad \{i,j\}\notin K,
\end{cases}
$$
then, the following statement describes conditions on $z_x$ guaranteeing compatibility of the distributions $(\mu_n)_{n\geq1}$
\begin{thm}\cite{BR}\label{t1}
Probability measures
$\mu^{(n)}$, $n=1,2,\ldots$, given by the formula (\ref{e1}),
are consistent if and only if for any $x\in V\setminus\{x^0\}$  the following equation holds:
\begin{equation}\label{e3}
z'_{i,x}=\prod_{y\in S(x)}{\sum_{j=0}^{m-1}a_{ij}\lambda_{i,j}z'_{j,y}+a_{im}\lambda_{i,m} \over \sum_{j=0}^{m-1}a_{mj}\lambda_{m,j}z'_{m,y}+a_{mm}\lambda_{m,m}}.
\end{equation}
where $z'_{i,x}=\lambda z_{i,x}/z_{m,x}, \quad i=0,1,\ldots,m-1$.
\end{thm}

More precisely, we denote $\theta=e^J$ and consider the activity $\lambda=(\lambda_{i,j})_{\{i,j\}\in K}$
defined as
$$\lambda_{i,j}=
\begin{cases}
1, \quad \mbox{if} \quad i=j\in 2\mathbb{Z},\\
\theta, \quad \mbox{if} \quad |i-j|=1,\\
0, \quad \mbox{otherwise}.
\end{cases}$$

For this activity and an even $m\in\mathbb{N}$, from (\ref{e3}) we obtain
\begin{equation}\label{e4}
\begin{cases}
z_{0,x}=\prod_{y\in S(x)}\frac{z_{0,y}+\theta z_{1,y}}{1+\theta z_{m-1,y}},\\[.1cm]
z_{1,x}=\prod_{y\in S(x)}\frac{\theta z_{0,y}+\theta z_{2,y}}{1+\theta z_{m-1,y}}, \\[.1cm]
z_{2i,y}=\prod_{y\in S(x)}\frac{\theta z_{2i-1,y}+z_{2i,y}+\theta z_{2i+1,y}}{1+\theta z_{m-1,y}}, \quad 1\leq i<m/2 \\[.1cm]
z_{2i+1,y}=\prod_{y\in S(x)}\frac{\theta z_{2i,y}+\theta z_{2i+2,y}}{1+\theta z_{m-1,y}}, \quad 0\leq i<m/2-1  \\[.1cm]
z_{m,x}=1
\end{cases}
\end{equation}
and, by Theorem \ref{t1}, for any $z=\{z_x:x\in V\}$ satisfying (\ref{e4}), there exists a unique SGM $\mu$ for the
HC-SOS model. However, the analysis of solutions to (\ref{e4}) for an arbitrary $m$ is challenging. We therefore
restrict our attention to a smaller class of measures, namely the translation-invariant SGMs.

When considering only translation-invariant measures, the functional equation (\ref{e4}) simplifies to:
\begin{equation}\label{e5}
	\begin{cases}
		z_{0}=\left(\frac{z_{0}+\theta z_{1}}{1+\theta z_{m-1}}\right)^k,\\[.1cm]
		z_{1}=\left(\frac{\theta z_{0}+\theta z_{2}}{1+\theta z_{m-1}}\right)^k, \\[.1cm]
		z_{2i}=\left(\frac{\theta z_{2i-1}+z_{2i}+\theta z_{2i+1}}{1+\theta z_{m-1}}\right)^k, \quad 1\leq i<m/2 \\[.1cm]
		z_{2i+1}=\left(\frac{\theta z_{2i}+\theta z_{2i+2}}{1+\theta z_{m-1}}\right)^k, \quad 0\leq i<m/2-1  \\[.1cm]
		z_{m}=1
	\end{cases}
\end{equation}

\section{Translation-invariant SGMs for the HC-SOS model with $m=2$}

In the following we restrict ourselves to the case $m=2$. In this case, denoting $x=\sqrt[k]{z_0}$
and $y=\sqrt[k]{z_1}$, from (\ref{e5}) we get
\begin{equation}\label{e6}
\begin{cases} x=\frac{x^k+\theta y^k}{1+\theta y^k},\\[2mm]
y=\frac{\theta x^k+\theta}{1+\theta y^k}
\end{cases}
\end{equation}

\begin{rk}
In \cite{BR}, this model is studied using the admissibility graph hinge. In the case of the hinge, $x^k$ is added to the numerator of the second equation in system (\ref{e6}). For the hinge graph with $k = 2, 3$ (while the case $k \geq 4$ remains unresolved), it is shown that as $\theta$ increases, the number of solutions representing translation-invariant Gibbs measures follows the sequence $ 1 \to 3 \to 5 \to 6 \to 7$.

In contrast, for our case, the wand graph, we demonstrate that {\bf for any} $k \geq 2$, the number of solutions is at most $3$. Furthermore, we derive an explicit formula for the critical value of $\theta$, denoted as $\theta_{\text{cr}}(k)$.
\end{rk}

Considering the first equation of system (\ref{e6}), we find the solutions $x=1$ and
\begin{equation}\label{e7}
	\theta y^k=x^{k-1}+x^{k-2}+\dots+x.
\end{equation}

We start by investigating the case $x=1$.

\textbf{Case $x=1$.} In this case, from the second equation in (\ref{e6}), we get that
\begin{equation}\label{e8}
f(y):=\theta y^{k+1}+y-2\theta=0
\end{equation}
and hence, the following is true.
\begin{lemma}
For any $k\geq2$, the equation (\ref{e8}) has unique positive solution.
\end{lemma}
\begin{proof}
It is easy to check that the function $f(y)$ is increasing, with $f(0)=-2\theta<0$ and $f(2\theta)=2^{k+1}\theta^{k+2}>0$. Hence, the equation (\ref{e8}) has a unique positive solution $y^*=y^*(k,\theta)$ for any $k\geq 2$.
\end{proof}

\textbf{Case $x\neq1$.} In this case, using (\ref{e7}) and the second equation in (\ref{e6}), we get
\begin{equation}\label{e9}
\theta^{k+1}=\eta(x):=\frac{\Big(\sum_{i=1}^{k-1}x^i\Big)\Big(\sum_{i=0}^{k-1}x^i\Big)^k}{(x^k+1)^k}.
\end{equation}
Note that if $x$ is a solution to (\ref{e9}), then $\frac{1}{x}$ is also a solution to (\ref{e9}). We shall show that the equation (\ref{e9}) has at most two positive solutions. To do this, consider the derivative of the function $\eta'(x)$
$$\eta'(x)=-\frac{(x^{k-1}+x^{k-2}+\dots+x)^{k}\cdot\vartheta(x,k)}{(x^{2k}-1)(x-1)^2(x^k+1)^k},$$
where
$$\vartheta(x,k)=(x^k-1)(x^{2k}-1)+kx(x^{k-2}-1)(x^{2k}-1)-2k^2x^k(x^{k-1}-1)(x-1).$$
It can be seen that $x=1$ is a two-fould root of the polynomial $\vartheta(x,k)$.
But in \cite{KMU} it was shown that $x=1$ is four-fold root of $\vartheta(x,k)$, i.e.,
$$\vartheta(t,k)=(x-1)^4\cdot\phi(x),$$
where $\phi(x)>0$ for $x>0$. It means that
$$\eta'(x)=-\frac{(x-1)^4(x^{k-1}+x^{k-2}+\dots+x)^{k}\cdot\phi(x)}{(x^{2k}-1)(x-1)^2(x^k+1)^k}.$$
Therefore, the function $\eta(x)$ increases for $x<1$, decreases for $x>1$, and reaches its maximum for $x=1$ (see Fig. \ref{fi2}):
$$\eta_{\max}=\eta(1)=\frac{(k-1)k^k}{2^{k}}.$$
It follows from (\ref{e9}) that
\begin{equation}\label{tcr} \theta_{cr}\equiv \theta_{cr}(k)=\sqrt[k+1]{\eta_{\max}}=\sqrt[k+1]{\frac{(k-1)k^k}{2^{k}}}.\end{equation}
\begin{figure}[h]
\center{\includegraphics[width=7cm]{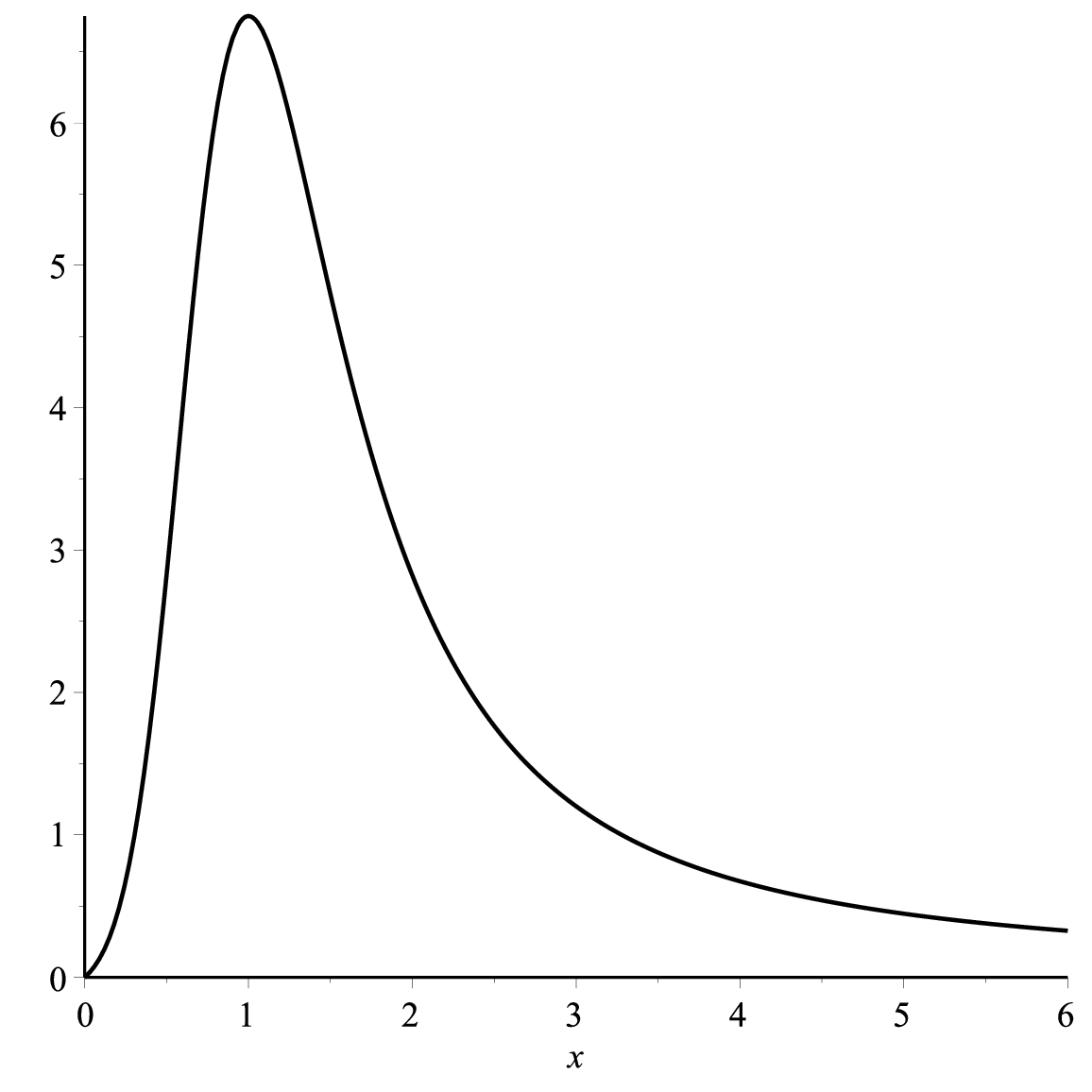}}
\caption{Graph of the function $\eta(x)$ for $k=3$.}\label{fi2}
\end{figure}

Hence, we have the following result:
\begin{itemize}
	\item for each fixed $\theta<\theta_{cr}$ there are two positive solutions to (\ref{e9}), one of them $>1$ and another  $<1$;
	\item if $\theta=\theta_{cr}$ then equation (\ref{e9}) has unique positive solution $x=1$;
	\item if $\theta>\theta_{cr}$ then the equation has no any positive solution.
\end{itemize}

	Then the system of equations (\ref{e6}) has a unique solution of the form $(1,y^*)$ for $\theta\geq\theta_{cr}$, and for $\theta<\theta_{cr}$ has exactly three solutions $(1,y^*)$, $(x^{(1)},y^{(1)})$ and $(x^{(2)},y^{(2)})$,  where $y^*$ is the unique positive solution of (\ref{e8}).

We can summarize our results for $k\geq2$ in the following
\begin{thm}\label{t2}
For the HC-SOS model in the case ``wand", with $m=2$ and $k\geq2$
the following assertions hold

1. If $\theta\geq\theta_{cr}$, then there is unique translation-invariant SGM, denoted by $\mu_0$.

2. If $\theta<\theta_{cr}$, then there are three translation-invariant SGMs, denoted as $\mu_i$, $i=0,1,2$.
Here $\theta_{cr}$ is defined in (\ref{tcr}).

\end{thm}
\section{Extremality of SGMs}

\subsection{Conditions for non-extremality of a SGM}

It is known that a translation-invariant SGM corresponding to a vector $v=(x, y)\in \mathbb{R}^2$ (which
is a solution to (\ref{e6})) is a tree-indexed Markov chain with states $\{0,1,2\}$, (see \cite{G}, Definition 12.2],
and the transition matrix
\begin{equation} \label{e10}
	\mathbb{P}(x,y)=
\begin{pmatrix}
  \frac{x^k}{x^k+\theta y^k} & \frac{\theta y^k}{x^k+\theta y^k} & 0 \\
  \frac{x^k}{1+x^k} & 0 &\frac{1}{1+x^k} \\
  0 & \frac{\theta y^k}{1+\theta y^k}& \frac{1}{1+\theta y^k}
\end{pmatrix}.
\end{equation}
A sufficient condition for non-extremality of a Gibbs measure $\mu$ corresponding to $\mathbb{P}(x,y)$ on a Cayley tree of order $k\geq2$ is given by the Kesten-Stigum condition $ks_2^2>1$, where $s_2$ is the
second-largest (in absolute value) eigenvalue of $\mathbb{P}$. 

In the following subsections, we will examine the Kesten-Stigum condition along with the extremality condition.

\subsection{Conditions for extremality of a SGM}

Let us first give some necessary definitions from \cite{MSW}. We consider the finite complete subtree $\mathcal T$,  containing all
initial points of the semitree $\Gamma^k_{x^0}$. The boundary $\partial \mathcal T$ of the subtree of its vertices, which are in  $\Gamma^k_{x^0}\setminus \mathcal T$. We identify the subtree $\mathcal T$ with the set of its vertices. The set of all
edges $A$ and $\partial A$ is denoted by $E(A)$.

In \cite{MSW}, the key parameters are  $\kappa$ and $\gamma$. They define the properties of Gibbs measures $\{\mu^\tau_{{\mathcal T}}\}$,  where the boundary condition $\tau$  is fixed and $\mathcal T$ is the arbitrary initial complete and finite subtree in  $\Gamma^k_{x^0}$.

For a given initial subtree $\mathcal T$ of the tree $\Gamma^k_{x^0}$ and the vertex  $x\in\mathcal T$ we write $\mathcal T_x$ for the (maximum) subtree $\mathcal T$ with the initial point at $x$.  If $x$ is not the initial point of the $\mathcal T$, then the
Gibbs measure is denoted by $\mu_{\mathcal T_x}^s$ where the ancestor $x$ has the spin $s$ and the configuration at the lower
boundary ${\mathcal T}_x$ (i.e. on $\partial {\mathcal T}_x\setminus \{x\}$) is given in terms of $\tau$.

The distance between two measures $\mu_1$ and $\mu_2$ on $\Omega$ is defined as
$$\|\mu_1-\mu_2\|_x={1\over 2}\sum_{i=0}^2|\mu_1(\sigma(x)=i)-\mu_2(\sigma(x)=i)|.$$

Let  $\eta^{x,s}$ be the configuration $\eta$ with the spin at $x$ equal to $s$.

Following \cite{MSW}, we define
$$\kappa\equiv \kappa(\mu)=\sup_{x\in\Gamma^k}\max_{x,s,s'}\|\mu^s_{{\mathcal T}_x}-\mu^{s'}_{{\mathcal T}_x}\|_x, \qquad \gamma\equiv\gamma(\mu)=\sup_{A\subset \Gamma^k}\max\|\mu^{\eta^{y,s}}_A-\mu^{\eta^{y,s'}}_A\|_x,$$
where the maximum is taken over all boundary conditions $\eta$,
all $y\in \partial A$, all neighbors  $x\in A$ of the vertex $y$, and all spins $s, s'\in \{1,\dots,q\}$.

A sufficient condition for the Gibbs measure $\mu$ to be extreme is the inequality (see Theorem 9.3 in \cite{MSW})
\begin{equation}\label{e12}
	k\kappa(\mu)\gamma(\mu)<1.
\end{equation}
Note that $\kappa$ has the simple form
$$
\kappa={1\over 2}\max_{i,j}\sum_{l=0}^2|P_{il}-P_{jl}|.
$$

\subsection{The case $\mu_0$}

In the case $x = 1$, denote matrix $\mathbb{P}_1=\mathbb{P}(1,y)$.

 The matrix $\mathbb{P}_1$ has three eigenvalues, $1$ and
$$\lambda_1(y,\theta,k)=-\frac{\theta y^k}{\theta y^k+1}, \qquad \lambda_2(y,\theta,k)=\frac{1}{\theta y^k+1}$$

Let $\mu_0$ denote the SGM corresponding to unique solution $(1,y^{*})$ of (\ref{e8}). The following lemma gives  $s_2$  the
second-largest eigenvalue of $\mathbb{P}_1$. 

\begin{lemma}\label{yy} For the solution $y^*$ of (\ref{e8}), we have that
 \begin{equation}\label{e11}
 	s_2=
 	\begin{cases}
 		\lambda_2(y^*,\theta,k), \quad \mbox{if} \quad 0<\theta\leq 1,\\
 		-\lambda_1(y^*,\theta,k), \quad \mbox{if} \quad \theta>1.
 	\end{cases}
 \end{equation}
\end{lemma}
\begin{proof}
	{\it Case: $0<\theta<1$.} In this case it is sufficient to show that $\theta y^k<1$ (here and in this proof $y=y^*$). From (\ref{e8}) we have that
  $$0<\theta=\frac{y}{2-y^{k+1}}<1.$$
  From this inequalities it follows that $y<\sqrt[k+1]{2}$ and
  $y^{k+1}+y<2$. The last inequality holds only for the case $y<1$. Therefore,  in the case $\theta\in(0,1)$
for any solution $y$ of (\ref{e8}) we have $y<1$, i.e., $\theta y^k<1$. Consequently,  $$|\lambda_1(y,\theta,k)|<\lambda_2(y,\theta,k).$$

{\it Case: $\theta=1$}. In this case $y=1$. Therefore
$$|\lambda_1(y,\theta,k)|=\lambda_2(y,\theta,k).$$

 {\it Case:  $\theta>1$.}  In this case we have   $$\theta=\frac{y}{2-y^{k+1}}>1.$$
This is true iff $y<\sqrt[k+1]{2}$ and
$y^{k+1}+y>2$, i.e., $1<y<\sqrt[k+1]{2}$. Therefore, for $\theta>1$, we have $\theta y^k>1$ and $$|\lambda_1(y,\theta,k)|>\lambda_2(y,\theta,k).$$
\end{proof}

For extremality of $\mu_0$ we need to calculate $\kappa$ and $\gamma$ corresponding to solution $(1,y)$.

It is clear that $|P_{il}-P_{jl}|=0$ for $i=j$. Using (\ref{e10}), for $i\neq j$ we calculate
$$\sum_{l=0}^2|P_{0l}-P_{2l}|=\frac{2}{1+\theta y^k}, \quad \sum_{l=0}^2|P_{0l}-P_{1l}|=\sum_{l=0}^2|P_{1l}-P_{2l}|=\frac{|1-\theta y^k|+3\theta y^k+1}{2(1+\theta y^k)}.$$

Consequently, by arguments of the proof of Lemma \ref{yy} we have
\begin{equation}\label{e13}
\kappa=
\begin{cases} \frac{1}{1+\theta y^k}, \quad \mbox{if} \quad 0<\theta<1,\\[2mm]
\frac{\theta y^k}{1+\theta y^k}, \quad \mbox{if} \quad \theta>1.
\end{cases}
\end{equation}
Now the estimate for $\gamma$, similar to, will be sought in the following form (see \cite{MSW})
$$
\gamma={1\over 2}\max\left\{\left\|\mu_{A}^{\eta^{y,0}}-\mu_{A}^{\eta^{y,1}}\right\|_x,\left\|\mu_{A}^{\eta^{y,1}}-\mu_{A}^{\eta^{y,2}}\right\|_x,
\left\|\mu_{A}^{\eta^{y,0}}-\mu_{A}^{\eta^{y,2}}\right\|_x\right\},
$$
where
$$\left\|\mu_{A}^{\eta^{y,0}}-\mu_{A}^{\eta^{y,1}}\right\|_x=\frac12\sum_{s=0}^2\left|\mu_{A}^{\eta^{y,0}}(\sigma(x)=s)-\mu_{A}^{\eta^{y,1}}(\sigma(x)=s)\right|=$$
$$=\frac{1}{2}\left(\left|P_{00}-P_{10}\right|+\left|P_{01}-P_{11}\right|+\left|P_{02}-P_{12}\right|\right)=\frac{|1-\theta y^k|+3\theta y^k+1}{4(1+\theta y^k)},$$
$$\left\|\mu_{A}^{\eta^{y,1}}-\mu_{A}^{\eta^{y,2}}\right\|_x=\frac12\sum_{l=0}^2\left|P_{1l}-P_{2l}\right|=
\frac{|1-\theta y^k|+3\theta y^k+1}{4(1+\theta y^k)},$$
$$\left\|\mu_{A}^{\eta^{y,0}}-\mu_{A}^{\eta^{y,2}}\right\|_x=\frac12\sum_{l=0}^2\left|P_{0l}-P_{2l}\right|=\frac{1}{1+\theta y^k}.$$
Therefore, for $\theta>0$, we obtain $\kappa=\gamma$. Then the condition (\ref{e12}) reads
\begin{equation}\label{e14}
k\kappa^2<1.
\end{equation}
Since $\kappa=\gamma=s_2$ (mentioned in Lemma \ref{yy}), comparing (\ref{e14}) with the Kesten-Stigum condition we obtain the following

\begin{pro}\label{yb} For measure $\mu_0$ to be extreme the Kesten-Stigum condition is sharp (except boundary value : $ks_2^2=1$).
\end{pro}

The following theorem is true
\begin{thm}\label{t3}
  Let $k=2$, $\theta_1=\frac12\sqrt[3]{4\sqrt{2}-4}$ and $\theta_2=\frac12\sqrt[3]{28+20\sqrt{2}}$. Then the measure $\mu_0$ is

  $\bullet$ non-extreme for $0<\theta<\theta_1$ or $\theta>\theta_2$,

   $\bullet$ extreme for $\theta_1<\theta<\theta_2$.
\end{thm}

\begin{proof}
First, let us establish the conditions under which a measure is non-extreme. For $k=2$ the unique positive solution of (\ref{e8}) is
$$y^*(2,\theta)=\frac{\sqrt[3]{3\theta\sqrt{81\theta^4+3\theta}+27\theta^3}}{3}-\frac{1}{\sqrt[3]{3\theta\sqrt{81\theta^4+3\theta}+27\theta^3}}.$$
From (\ref{e11}) we get $s_2=\lambda_2(y,\theta,2)=\frac{1}{\theta y^2+1}$ for $\theta\in(0,1)$. Then the Kesten-Stigum condition is of the form
$$\frac{2}{(\theta y^2(\theta)+1)^2}>1.$$
 Consider the following function
$$h_2(\theta):=\frac{2}{(\theta y^2(\theta)+1)^2}-1.$$
Note that, the function $h_2(\theta)$ decreases, i.e.,
$$h_2'(\theta)=\frac{108\sqrt[3]{3}\theta a(\sqrt[3]{3}\theta-a^2)\left(3\sqrt[6]{243}\theta^3+(3\sqrt{3}\theta^2+\sqrt{27\theta^4+\theta})a^2+\sqrt[3]{3}\theta\sqrt{27\theta^4+\theta}\right)}
{\sqrt{27\theta^4+\theta}\left((3\sqrt[3]{9}\theta^2+\sqrt[6]{3}\sqrt{27\theta^4+\theta})a+\sqrt[3]{3}\theta+a^2\right)^3}<0,$$
where $a=\sqrt[3]{\theta\sqrt{81\theta^4+3\theta}+9\theta^2}$.
Indeed, inequality $h_2'(\theta)<0$ is equivalent to $a^2-\sqrt[3]{3}\theta>0$. After some algebras, we have
 $$a^6-3\theta^3=162\theta^6+18\theta^4\sqrt{81\theta^4+3\theta}>0.$$
  On the other hand $h_2(0.2)\approx0.8846$ and $h_2(1)=-\frac12$.
Consequently, the equation $h_2(\theta)=0$ has unique positive solution:
 $$\theta=\theta_1=\frac12\sqrt[3]{4\sqrt{2}-4}.$$

Then, the measure $\mu_0$ is non-extreme when $h_2(\theta)>0$, i.e.,  $0<\theta<\frac12\sqrt[3]{4\sqrt{2}-4}\approx0.5916$ (see Fig. \ref{fi3}).

For $\theta>1$, from (\ref{e11}) we get $s_2=\lambda_1(y,\theta,2)=-\frac{\theta y^2}{\theta y^2+1}$. Then the Kesten-Stigum condition reduced to
$$2\cdot\left(\frac{\theta y^2(\theta)}{\theta y^2(\theta)+1}\right)^2>1.$$
Consider the following function
$$q_2(\theta):=2\cdot\left(\frac{\theta y^2(\theta)}{\theta y^2(\theta)+1}\right)^2-1.$$

Note that, the function $q_2(\theta)$ increases, i.e.,
$$q_2'(\theta)=\frac{36\sqrt[3]{3}(a^2-\sqrt[3]{3}\theta)^3\left(3\sqrt[6]{243}\theta^3+(3\sqrt{3}\theta^2+\sqrt{27\theta^4+\theta})a^2+\sqrt[3]{3}\theta\sqrt{27\theta^4+\theta}\right)}
{a\sqrt{27\theta^4+\theta}\left((3\sqrt[3]{9}\theta^2+\sqrt[6]{3}\sqrt{27\theta^4+\theta})a+\sqrt[3]{3}\theta+a^2\right)^3},$$
where $a=\sqrt[3]{\theta\sqrt{81\theta^4+3\theta}+9\theta^2}$.

Then, inequality $q_2'(\theta)>0$ is equivalent to $a^2-\sqrt[3]{3}\theta>0$.
 It follows that, the function $q_2(\theta)$ increases. On the other hand $q_2(1)=-\frac12$ and $q_2(4)=0.4476$.
Hence the equation $q_2(\theta)=0$ has unique positive solution:  $\theta=\theta_2=\frac12\sqrt[3]{28+20\sqrt{2}}$.

Then, the measure $\mu_0$ is non-extreme if $q_2(\theta)>0$, i.e., $\theta>\frac12\sqrt[3]{28+20\sqrt{2}}\approx1.9161$ (see Fig. \ref{fi3}).
Thus we have proved the first part of theorem, the second part, i.e., condition of extremality, follows by Proposition \ref{yb}.

%
%
\begin{figure}[h]
\center{\includegraphics[width=6cm]{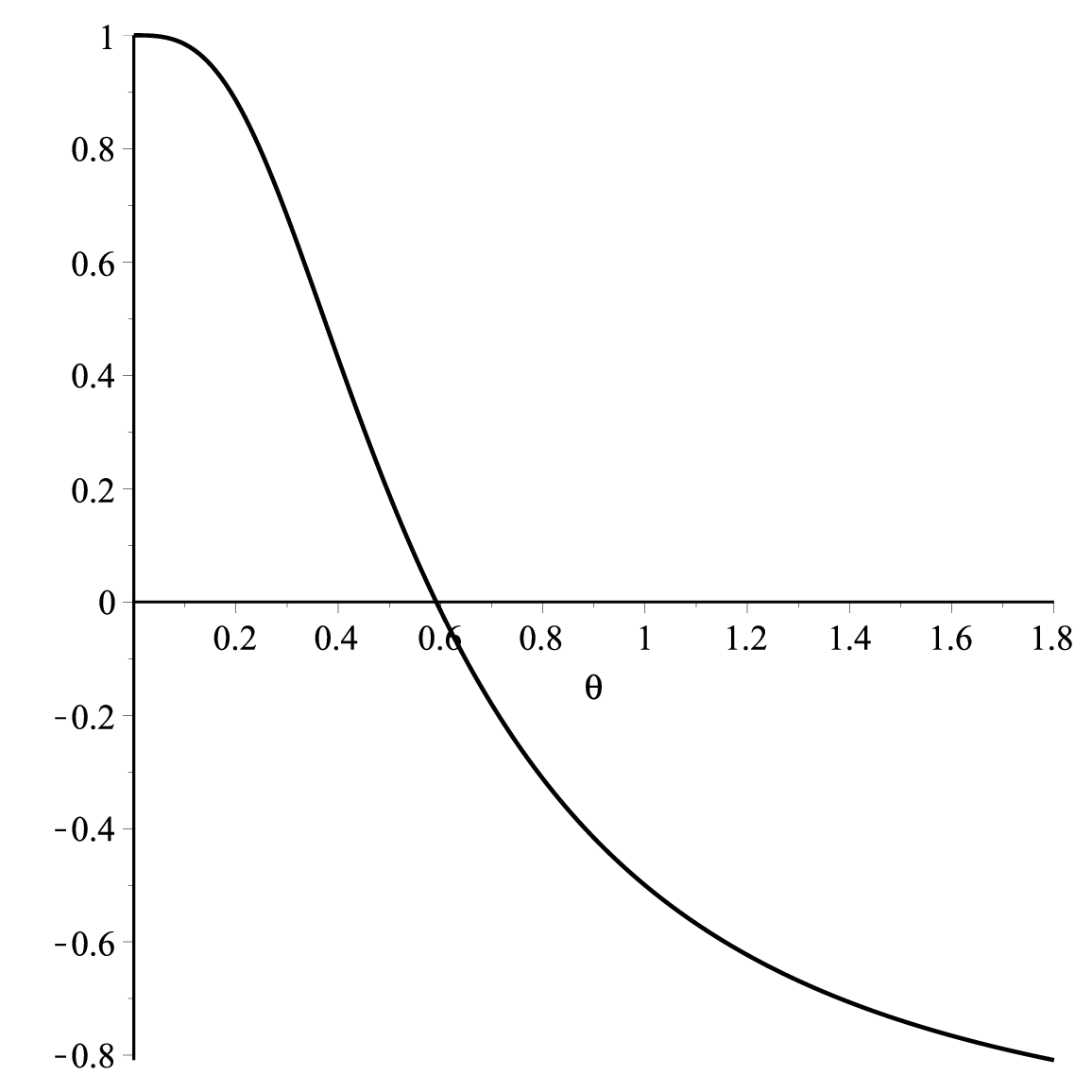} \qquad \qquad \includegraphics[width=6cm]{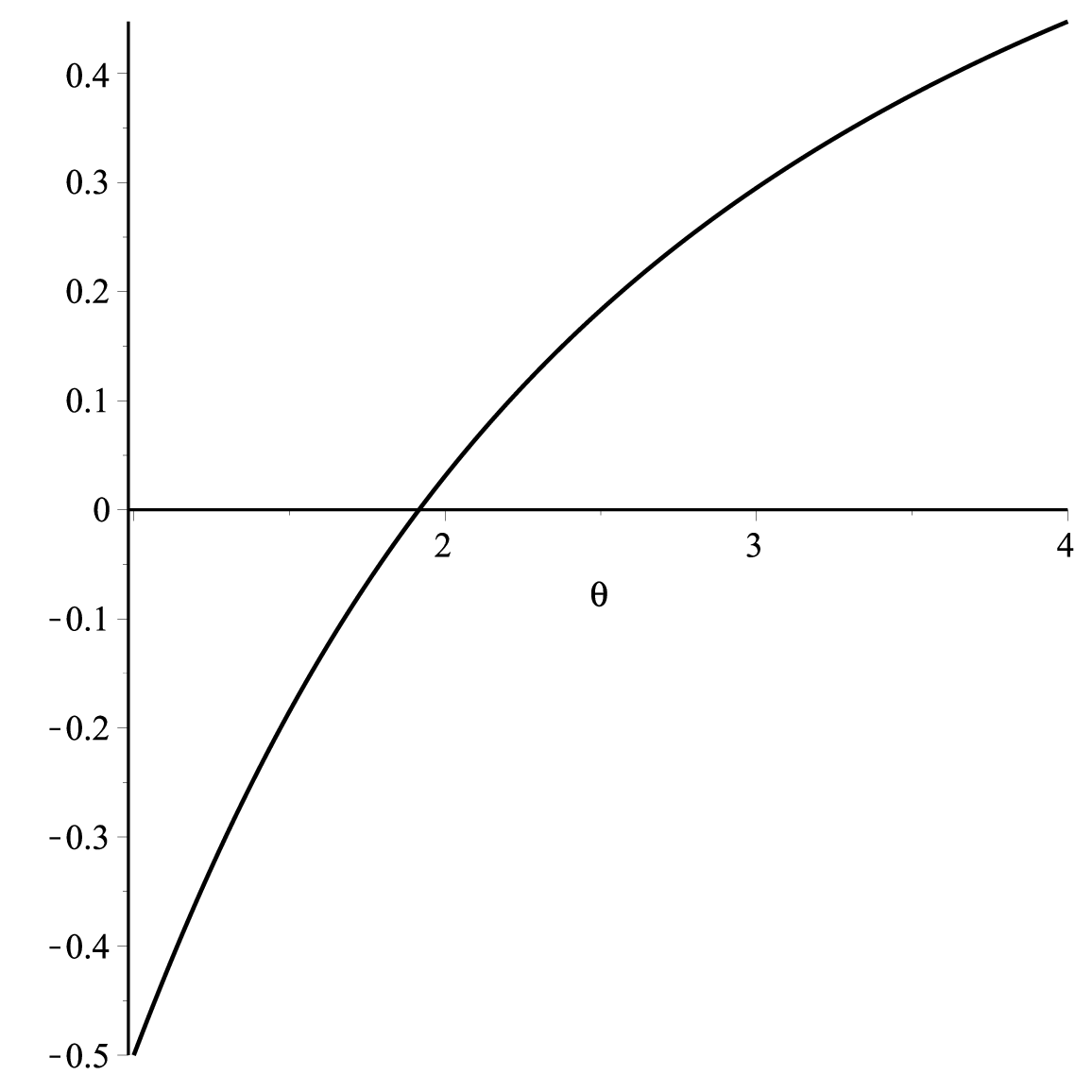} }
\caption{Graph of the function $h_2(\theta)$ (left) and $q_2(\theta)$ (right).}\label{fi3}
\end{figure}
\end{proof}

\begin{thm} \label{t4}
  Let $k=3$. There are $\theta_3\approx0.801$ and $\theta_4\approx1.8462$ such that the measure $\mu_0$ is

   $\bullet$ non-extreme for $0<\theta<\theta_3$ or $\theta>\theta_4$,

   $\bullet$ extreme for $\theta_3<\theta<\theta_4$.
\end{thm}

\begin{proof}
First, let us establish the conditions under which a measure is non-extreme. For $k=3$ the unique positive solution of (\ref{e8}) is
$$y^*(3,\theta)=\frac1{\sqrt[4]{8\theta^3 d(\theta)}}-\sqrt{\frac{\theta d(\theta)}{2}},$$
where $$d(\theta)=\frac{\sqrt[3]{12\theta\sqrt{6144\theta^4+81}+324\theta}}{12\theta}-\frac{8\theta}{\sqrt[3]{12\theta\sqrt{6144\theta^4+81}+324\theta}}.$$
From (\ref{e11}) we get $s_2=\lambda_2(y,\theta,3)=\frac{1}{\theta y^3+1}$ for  $\theta\in(0,1)$. Then the Kesten-Stigum condition is reduced to
$$h_3(\theta):=\frac{3}{(\theta y^3(\theta)+1)^2}-1>0.$$
It can be seen from the graph of $h_3(\theta)$ that $\mu_0$ is non-extreme for $0<\theta<\theta_3\approx0.801$ (see Fig. \ref{fi4}).

For $\theta>1$, we get $\lambda=\lambda_1(y,\theta,3)=-\frac{\theta y^3}{\theta y^3+1}$. Then the Kesten-Stigum condition is reduced to
$$q_3(\theta):=3\cdot\left(\frac{\theta y^3(\theta)}{\theta y^3(\theta)+1}\right)^2-1>0.$$
It can be seen from the graph of $q_3(\theta)$ that $\mu_0$ is non-extreme for $\theta>\theta_4\approx1.8462$ (see Fig. \ref{fi4}).

%

\begin{figure}[h]
\center{\includegraphics[width=4.5cm]{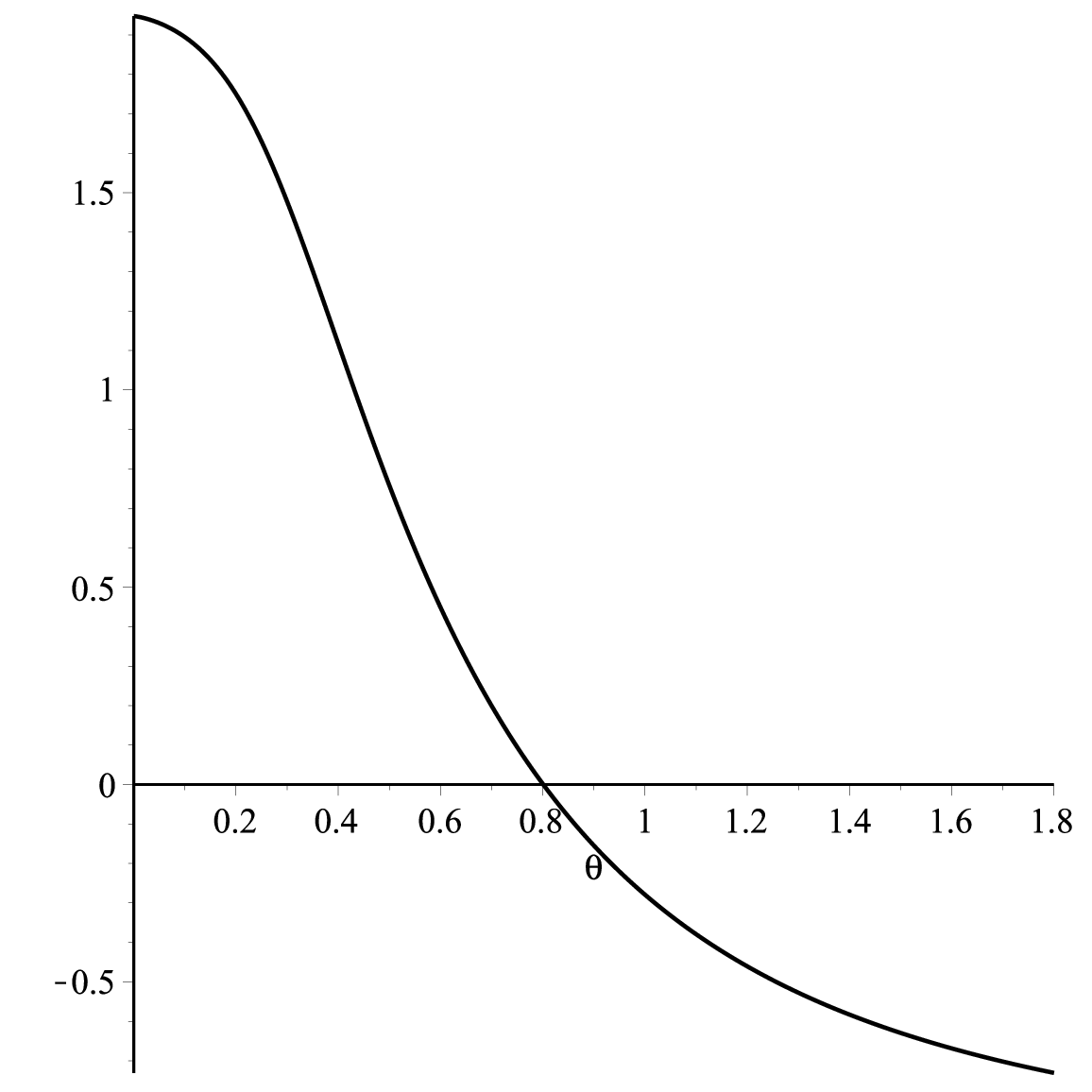}\quad \includegraphics[width=4.5cm]{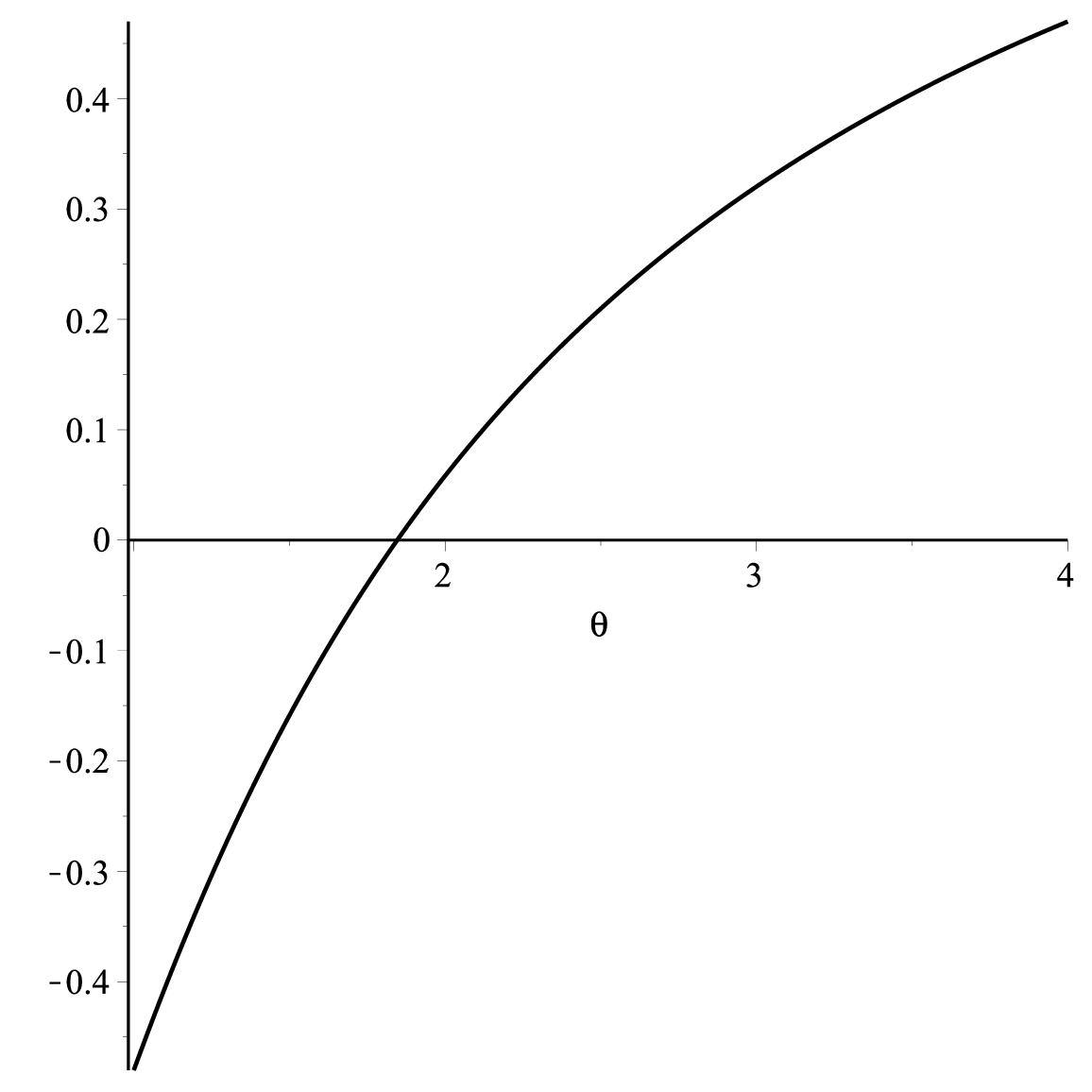} \quad \includegraphics[width=4.5cm]{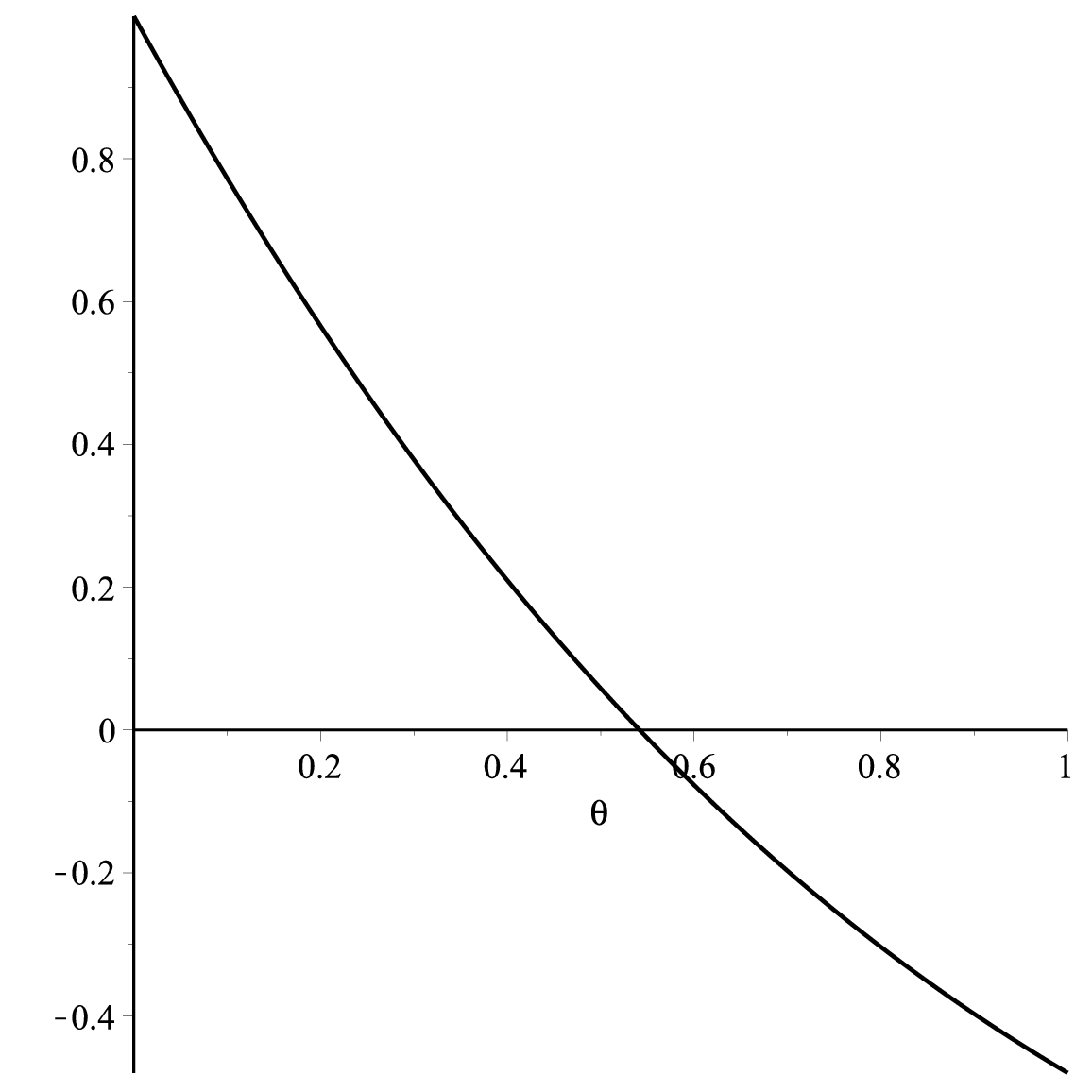}}
\caption{Graph of the function $h_3(\theta)$ (left), $q_3(\theta)$ (middle) and $q_3\left(\frac{1}{\theta}\right)$ (right).}\label{fi4}
\end{figure}
The extremality part follows by Proposition \ref{yb}.
\end{proof}

\begin{thm} \label{t5}
 Let $k\geq4$. Then the measure $\mu_0$ is non-extreme for any $\theta>0$.
\end{thm}

\begin{proof}
\textbf{Case $\theta\in(0,1)$}.
In this case, we have $s_2=\lambda_2(y,\theta,k)=\frac{1}{\theta y^k+1}$. Then the Kesten-Stigum condition reads
$$\frac{k}{(\theta y^k+1)^2}>1.$$
By proof of Lemma \ref{yy} we know that if $\theta\in(0,1)$ then $\theta y^k<1$ therefore
$$\frac{k}{(\theta y^k+1)^2}>\frac{k}{4}\geq1.$$
It follows that $\mu_0$ is non-extreme for $k\geq4$.

\textbf{Case $\theta>1$}.
In this case, we have $s_2=\lambda_1(y,\theta,k)=-\frac{\theta y^k}{\theta y^k+1}$. Then the Kesten-Stigum condition is reduced to
$$k\cdot\left(\frac{\theta y^k}{\theta y^k+1}\right)^2>1.$$
From (\ref{e8}) we have that
\begin{equation}\label{e15}
\theta=\frac{y}{2-y^{k+1}}
\end{equation}
  Then the Kesten-Stigum condition is equivalent to $ky^{2k+2}>{4}$. From (\ref{e15}), we get $y>1$ for $\theta>1$.
Then, the above inequality is fulfilled in $k\geq4$. It follows that $\mu_0$ is non-extreme for $k\geq4$.
\end{proof}

\subsection{Extremality of SGMs $\mu_1$ and $\mu_2$}

In this section, we identify conditions on the model parameters that ensure the (non-) extremality of the measures $\mu_1$ and $\mu_2$.
To study the (non-)extremality of the measures $\mu_1$ and $\mu_2$, for $x\neq1$ we need explicit solutions of the system (\ref{e6}) corresponding to these measures.

For $k=2$ we note that the measures $\mu_1$ and $\mu_2$ exist for $\theta<\theta_{cr}(2)=1$. In this case from (\ref{e7}) we have
$$\theta y^2=x.$$
Using the last equality, we write the second equality (\ref{e6}) for $k=2$ as
\begin{equation}\label{e16}
x=\theta^3\left(\frac{x^2+1}{1+x}\right)^2.
\end{equation}
From (\ref{e16}) we obtain
\begin{equation}\label{e17}
\theta^3x^4-x^3+2(\theta^3-1)x^2-x+\theta^3=0.
\end{equation}
In (\ref{e17}) we introduce the notation $\rho=x+\frac{1}{x}$. Then
\begin{equation}\label{e18}
\theta^3\rho^2-\rho-2=0.
\end{equation}
Here $\rho>2$ since $x\neq 1$. Solutions of (\ref{e18}) have the form
$$\rho_1=\frac{1+\sqrt{1+8\theta^3}}{2\theta^3}, \quad \rho_2=\frac{1-\sqrt{1+8\theta^3}}{2\theta^3}.$$
It is not difficult to see that  $\rho_1>2$ and $\rho_2<0$ for $0<\theta<1$.

Let $\rho=\rho_1$. Then from $x+\frac{1}{x}=\rho$ we get
\begin{equation}\label{e19}
x_1=\frac{\rho+\sqrt{\rho^2-4}}{2}, \quad x_2=\frac{\rho-\sqrt{\rho^2-4}}{2}.
\end{equation}
From $\theta y^2=x$ we find the corresponding values $y_i$:
$$y_1=\sqrt{\frac{\rho+\sqrt{\rho^2-4}}{2\theta}}, \quad y_2=\sqrt{\frac{\rho-\sqrt{\rho^2-4}}{2\theta}}.$$
It is clear that $x_1>1$, $x_2<1$ and $x_1\cdot x_2=1$ for $0<\theta<1$.

{\bf Non-extremality of SGMs $\mu_1$ and $\mu_2$.}

In the case $k=2$ using the equality $\theta y^2=x$ we find the matrix of probability transitions for solutions $(x_1;y_1)$ and $(x_2;y_2)$
\begin{equation} \label{e20}
	\mathbb{P}_2(x)=
\begin{pmatrix}
  \frac{x}{x+1} & \frac{1}{x+1} & 0 \\[2mm]
  \frac{x^2}{1+x^2} & 0 &\frac{1}{1+x^2} \\[2mm]
  0 & \frac{x}{1+x}& \frac{1}{1+x}
\end{pmatrix}.
\end{equation}
The matrix $\mathbb{P}_2(x)$ has three eigenvalues, 1 and
$$\lambda_{1}=\frac{\sqrt{2} x}{(x+1)\sqrt{x^2+1}}, \quad \lambda_{2}=-\frac{\sqrt{2} x}{(x+1)\sqrt{x^2+1}}.$$
Therefore,
\begin{equation}\label{ms} s_2=\frac{\sqrt{2} x}{(x+1)\sqrt{x^2+1}}.
	\end{equation}

\begin{rk} \label{r2}
In the formula for  $s_2$  presented in (\ref{ms}), it is evident that the value of $s_2$  remains invariant when  $x$  is replaced by $1/x$. Consequently, it suffices to examine the non-extremality of either measure  $\mu_1$  or $\mu_2$.

\end{rk}

First, let us establish the conditions under which a measure is non-extreme. The condition of Kesten-Stigum has the form
$$\frac{4x^2}{(x^2+1)(x+1)^2}>1.$$
From here we have
$$x^4+2x^3-2x^2+2x+1=(x-1)^2+2(x^2+1)<0.$$
This means that the Kesten-Stigum condition is not satisfied. Hence, the measures $\mu_1$ and $\mu_2$ should be
extreme, which we shall check below.\\

{\bf Conditions for extremality}.

Using (\ref{e20}), for $i\neq j$ we calculate
$$\sum_{l=0}^2|P_{0l}-P_{2l}|=\frac{x+1+|x-1|}{x+1}, \quad \sum_{l=0}^2|P_{0l}-P_{1l}|=\frac{x^2+x+2+|x^2-x|}{(x+1)(x^2+1)},$$
$$\sum_{l=0}^2|P_{1l}-P_{2l}|=\frac{2x^3+x^2+x+|x^2-x|}{(x+1)(x^2+1)}.$$
Then we get
\begin{equation}\label{e21}
\kappa=
\begin{cases} \frac{x^2}{x^2+1}, \quad \mbox{if} \quad x>1,\\[2mm]
\frac{1}{x^2+1}, \quad \mbox{if} \quad x<1.
\end{cases}
\end{equation}

From the expression for $\kappa$ (see also Remark \ref{r2})  it follows that it is sufficient to study the extremality of one of the measures $\mu_1$ and $\mu_2$.

Now the estimate for $\gamma$, similar to, will be sought in the following form
$$
\gamma={1\over 2}\max\left\{\left\|\mu_{A}^{\eta^{y,0}}-\mu_{A}^{\eta^{y,1}}\right\|_x,\left\|\mu_{A}^{\eta^{y,1}}-\mu_{A}^{\eta^{y,2}}\right\|_x,
\left\|\mu_{A}^{\eta^{y,0}}-\mu_{A}^{\eta^{y,2}}\right\|_x\right\},
$$
where
$$\left\|\mu_{A}^{\eta^{y,0}}-\mu_{A}^{\eta^{y,1}}\right\|_x=\frac12\sum_{s=0}^2\left|\mu_{A}^{\eta^{y,0}}(\sigma(x)=s)-\mu_{A}^{\eta^{y,1}}(\sigma(x)=s)\right|=\frac{x^2+x+2+|x^2-x|}{(x+1)(x^2+1)}$$
$$\left\|\mu_{A}^{\eta^{y,1}}-\mu_{A}^{\eta^{y,2}}\right\|_x=\frac12\sum_{l=0}^2\left|P_{1l}-P_{2l}\right|=
\frac{2x^3+x^2+x+|x^2-x|}{(x+1)(x^2+1)},$$
$$\left\|\mu_{A}^{\eta^{y,0}}-\mu_{A}^{\eta^{y,2}}\right\|_x=\frac12\sum_{l=0}^2\left|P_{0l}-P_{2l}\right|=\frac{x+1+|x-1|}{x+1}.$$
Therefore, for $\theta>0$, we obtain $\kappa=\gamma$. Then the condition (\ref{e12}) reads
\begin{equation}\label{e22}
2\kappa^2<1.
\end{equation}

The following theorem is true
\begin{thm}\label{t6}
Let $k=2$ and $\theta_5=\sqrt[3]{\frac{2+\sqrt{2+2\sqrt{2}}}{2+2\sqrt2}}\approx 0.954$. Then the measures $\mu_1$  and $\mu_2$ are
extreme for $\theta_5<\theta<1$.
\end{thm}
\begin{proof}
Let $\kappa=\frac{x^2}{x^2+1}$ for $x_1$. Then from (\ref{e22}) we have
$$\frac{2x^4}{(x^2+1)^2}<1.$$
Hence $x_1<\sqrt{1+\sqrt{2}}$, i.e.
$$\frac{\rho+\sqrt{\rho^2-4}}{2}<\sqrt{1+\sqrt{2}}.$$
The solution to the last inequality is $\rho>\sqrt{2+2\sqrt{2}}$ or
$$\frac{1+\sqrt{1+8\theta^3}}{2\theta^3}>\sqrt{2+2\sqrt{2}}.$$
Hence
$\sqrt[3]{\frac{2+\sqrt{2+2\sqrt{2}}}{2+2\sqrt2}}<\theta<1$.
\end{proof}

Theorems 3 and 4 imply the following corollaries 1 and 2, respectively.

\begin{cor}
Let $k=2$. If $\theta_1< \theta<\theta_5$ then for HC-SOS model in the case ``wand", with $m=2$ there are at least two extreme Gibbs measures.
\end{cor}

\begin{proof} From Theorem 2 it is known that for any $\theta>0$ there exists a unique translation-invariant Gibbs measure $\mu_0$. By Theorem 3, for $\theta_1< \theta<\theta_2$ the measure $\mu_0$ is extreme. For $\theta_1< \theta<\theta_{cr}(2)=1$ we have measure $\mu_0$ and at least two new measures $\mu_1$ and $\mu_2$ mentioned in Theorem 2. By Theorem 6 the measures $\mu_1$  and $\mu_2$ are extreme for $\theta_5<\theta<1$. If we assume that all the new measures are not extreme for $\theta_1< \theta<\theta_5$ then it remains only one extreme measure $\mu_0$. But in this case the non-extreme measures can not be decomposed by the unique measure $\mu_0$. Consequently, at least one of the new measures must be extreme.
\end{proof}

\begin{cor}
Let $k=3$. If $\theta_3< \theta<\theta_{cr}(3)$ then for HC-SOS model in the case ``wand", with $m=2$ there are at least two extreme Gibbs measures.
\end{cor}

\begin{proof}
It is proved similarly to the proof of Corollary 1, using Theorem 4.
\end{proof}

\section*{ Acknowledgements}

The work supported by the fundamental project (number: F-FA-2021-425)  of The Ministry of Innovative Development of the Republic of Uzbekistan.

\end{document}